\documentclass[11pt]{article}

\usepackage{adjustbox}
\usepackage{algorithm}
\usepackage{algorithmicx}
\usepackage{algpseudocode}
\usepackage{amsthm}
\usepackage{amscd}
\usepackage{amsmath}
\usepackage{amssymb}
\usepackage{amsfonts}
\usepackage{bbm}
\usepackage{bm}
\usepackage{mathabx}
\usepackage{mathbbol}
\usepackage{mathrsfs}
\usepackage{mathtools}
\usepackage{stmaryrd}

\usepackage{booktabs}
\usepackage{color, colortbl}
\usepackage{dsfont}
\usepackage{enumerate}
\usepackage{filecontents}
\usepackage{graphicx}
\usepackage{lipsum}
\usepackage{multirow}
\usepackage{tikz-cd}
\usepackage{wrapfig}
\usepackage{subcaption}

\usepackage{setspace}
\usepackage{url}
\usepackage{geometry}

\usepackage{natbib}
\usepackage{hyperref}
\usepackage{cleveref} 
\usepackage{xcolor}

\newcommand\myshade{85}
\definecolor{YellowOrange}{RGB}{255,174,66}
\definecolor{Aquamarine}{RGB}{127,255,212}
\colorlet{mylinkcolor}{violet}
\colorlet{mycitecolor}{YellowOrange}
\colorlet{myurlcolor}{Aquamarine}
\hypersetup{
  linkcolor  = mylinkcolor!\myshade!black,
  citecolor  = myurlcolor!\myshade!black,
  urlcolor   = myurlcolor!\myshade!black,
  colorlinks = true
}

\usetikzlibrary{shapes,decorations,arrows,calc,arrows.meta,fit,positioning}
\tikzset{
    -Latex,auto,node distance =1 cm and 1 cm,semithick,
    state/.style ={ellipse, draw, minimum width = 0.7 cm},
    point/.style = {circle, draw, inner sep=0.04cm,fill,node contents={}},
    bidirected/.style={Latex-Latex,dashed},
    el/.style = {inner sep=2pt, align=left, sloped}
}

\newcommand{\R}{\mathbb{R}}
\newcommand{\Prob}{\mathbb{P}}
\newcommand{\E}{\mathbb{E}}

\newcommand{\indep}{\rotatebox[origin=c]{90}{$\models$}}

\newcommand{\overbar}[1]{\mkern 1.5mu\overline{\mkern-1.5mu#1\mkern-1.5mu}\mkern 1.5mu}
\newcommand{\dx}[1]{\ \text{d} #1}

\newcommand{\rW}{\overbar{W}}
\newcommand{\rX}{\overbar{X}}
\newcommand{\rY}{\overbar{Y}}
\newcommand{\rw}{\overbar{w}}
\newcommand{\rx}{\overbar{x}}
\newcommand{\ry}{\overbar{y}}

\newcommand{\fF}{\mathscr{F}}

\newcommand{\rZ}{\overbar{Z}}

\newcommand{\abs}[1]{\left|{#1}\right|}

\theoremstyle{plain}
\newtheorem{thm}{Theorem}

\newtheorem{cor}{Corollary}

\theoremstyle{definition}

\newtheorem{assum}{Assumption}

\theoremstyle{remark}
\newtheorem{remark}{Remark}

\allowdisplaybreaks

\newcommand{\blind}{0}

\title{The role of discretization scales in causal inference with continuous-time treatment}

\if1\blind
{
  \author{[Blinded for review]}
} \fi

\if0\blind
{
\author{Jinghao Sun$^{1}$, 
and Forrest W. Crawford$^{1,2,3,4}$ \\[1em]
1. Department of Biostatistics, Yale School of Public Health \\
2. Department of Statistics \& Data Science, Yale University \\
3. Department of Ecology \& Evolutionary Biology, Yale University \\
4. Yale School of Management \\
} 
} \fi

\begin{document}

\maketitle

\begin{abstract}
\noindent There are well-established methods for identifying the causal effect of a time-varying treatment applied at discrete time points. However, in the real world, many treatments are continuous or have a finer time scale than the one used for measurement or analysis. While researchers have investigated the discrepancies between estimates under varying discretization scales using simulations and empirical data, it is still unclear how the choice of discretization scale affects causal inference.
To address this gap, we present a framework to understand how discretization scales impact the properties of causal inferences about the effect of a time-varying treatment. We introduce the concept of ``identification bias", which is the difference between the causal estimand for a continuous-time treatment and the purported estimand of a discretized version of the treatment. We show that this bias can persist even with an infinite number of longitudinal treatment-outcome trajectories.
We specifically examine the identification problem in a class of linear stochastic continuous-time data-generating processes and demonstrate the identification bias of the g-formula in this context. Our findings indicate that discretization bias can significantly impact empirical analysis, especially when there are limited repeated measurements. Therefore, we recommend that researchers carefully consider the choice of discretization scale and perform sensitivity analysis to address this bias.
We also propose a simple and heuristic quantitative measure for sensitivity concerning discretization and suggest that researchers report this measure along with point and interval estimates in their work. By doing so, researchers can better understand and address the potential impact of discretization bias on causal inference.
\\[1em]

\end{abstract}


\section{Introduction}

Under the potential outcome framework \citep{rubin2005causal}, causal inference for a time-varying treatment is challenging due to time-varying confounding and treatment-confounder (outcome) feedback over time \citep[Chap.~19,~20]{hernan2020july}. Na\"ive adjustments for time-varying confounders can result in biased estimates of causal effects, particularly when these confounders are also mediators of the effect of prior treatment on the outcome. To address this challenge, ``g-methods" \citep[Chap.~21]{naimi2017introduction, hernan2020july}, such as the g-formula, marginal structural models, and structural nested models, have emerged as effective tools for identifying causal effects of time-varying treatments. These methods have been widely applied in empirical research in medicine and public health \citep{saul2019downstream, vangen2018hypothetical, arabi2018corticosteroid, naimi2013causal}.

It is widely acknowledged that g-methods are effective for analyzing regular longitudinal data, where each unit is measured at the same time points, also known as the \emph{discretization grid} or \emph{observation plan}, for all variables. However, in reality, many variables of epidemiological, medical, and social interest are continuous trajectories that can evolve over time. Examples of such trajectories include opioid doses \citep{straub2022trajectories}, physical activities \citep{mok2019physical}, body mass index (BMI) \citep{buscot2018bmi}, and socioeconomic status \citep{gustafsson2010life}, which may change at much finer time scales. In such cases, the treatment-outcome process may be better regarded as a continuous-time stochastic process, in which time-varying confounding can occur continuously.

Continuous-time processes do not naturally lend themselves to discretization, and the observed discretization is often chosen arbitrarily due to limitations in data collection or intentional decisions made by researchers to simplify analysis. For instance, in studies of medical interventions, the number and timing of clinic visits are often dictated by practical or budgetary considerations. Researchers may then proceed to analyze the data by selecting a discretization grid to facilitate analysis. Some common scenarios include using a coarser discretization to reduce computational costs \citep[e.g.][]{neugebauer2014targeted}, or selecting an arbitrary discretization grid to apply g-methods with irregular longitudinal data from electronic health record (EHR) data \citep[e.g.][]{pullenayegum2022randomized, katsoulis2021weight}.

The na\"ive application of g-methods to discretely sampled continuous-time trajectories can introduce an ``identification bias". This bias may persist even as the number of observed treatment-outcome trajectories approaches infinity. The identification bias is the difference between the discrete g-method functional and the true causal estimand, such as the counterfactual mean. While the nature of this bias remains incompletely understood, recent research indicates that the choice of discretization scales can significantly affect causal estimates for the effect of time-varying treatments. To address this, \citet{hernan2009observation} rely on the critical assumption of ``coincidence between observation times and times of potential treatment change" for more transparency. In contrast, \citet{etievant2021causal} investigate the extreme case of the coarsest discretization, where only one time point is available, and demonstrate that only under stringent assumptions can the g-formula estimand using the coarsest data be expressed as a weighted average of counterfactual means of a set of longitudinal treatment plans. 
Sensitivity analysis conducted by \citet{sofrygin2019targeted} illustrate differences in estimates of targeted minimum loss-based estimation (TMLE) under four discretization scales using EHR data. \citet{adams2020impact} and \citet{ferreira2020impact} assume the existence of a finest discrete-time data generating process (DGP) and explore the impact of observing only a coarser subset of the finest grid. Using simulations and empirical EHR data sets, they demonstrate that common methods, including the g-formula, IPW, and TMLE, may exhibit large biases when the discretization grid is coarse.

In this paper, we study the effect of discretization scales on causal estimands for time-varying treatments when the true causal data-generating process (DGP) is continuous in time.  We first introduce the potential outcome notation to continuous time \citep{sun2022causal} and highlight the conceptual issues that arise from discretization in causal inference with continuous-time DGPs. We then identify the two constituent parts of identification bias, which we call the ``discretization bias" and the ``asymptotic bias". 
We provide an in-depth study of a class of linear continuous-time stochastic DGPs and illustrate how the identification bias is influenced by discretization.
We further propose a simple and heuristic quantitative measure for sensitivity concerning discretization, and suggest researchers report this measure along with point and interval estimates in studies.


\section{Review of discrete-time causal identification}

Consider a longitudinal study with measurements at $J+1$ discrete time points, denoted by $W_k$, $Y_k$, and $Z_k$ for the observable treatment, outcome, and covariates at the $k$th time point, where $k$ ranges from $0$ to $J$. We define $X_k$ as the pair of $(Y_k, Z_k)$. Additionally, we define $\rW_k$, $\rX_k$, and $\rY_k$ as the history $(W_0, \ldots, W_k)$, $(X_0, \ldots, X_k)$, and $(Y_0, \ldots, Y_k)$, respectively. The observable data is comprised of $\mathcal{D}^J \equiv \{\rW_J, \rX_J\}$.

For a deterministic static treatment plan $\rw_{J-1}$ that specifies the treatment values up to the $(J-1)$th time point, we denote the corresponding potential outcome at time point $J$ as $Y_J^{\rw_{J-1}}$. This is the outcome value at the end of the study if the individual had followed the treatment plan $\rw_{J-1}$. The causal estimand is the counterfactual mean $\E[Y_J^{\rw_{J-1}}]$.

We focus on the g-formula identification strategy, which is nonparametrically equivalent to the inverse probability weighting (IPW) strategy with a specified marginal structural model (MSM) \citep{robins1999association}. To identify $\E[Y_J^{\rw_{J-1}}]$ with the g-formula in discrete-time, three identification assumptions are needed \citep{hernan2020july}. We define $\text{Supp}(U)$ as the support of a random variable or vector $U$. For notational convenience, we let $W_{-1} = X_{-1} = 0$.

\begin{assum}[Sequential Consistency (SC)]
    \label{assum:SC}
    \[X^{\rw_{J-1}}_{J} = X_{J}, \text{ if } \rw_{J-1} =  \rW_{J-1}.\]
\end{assum}

\begin{assum}[Sequential Positivity (SP)]
    \label{assum:SP}
For $k = 0, \ldots, J-1$ and all $(\rx_k, \rw_{k-1}) \in \text{Supp}(\rX_k, \rW_{k-1})$, $(\rx_k, \rw_{k}) \in \text{Supp}(\rX_k, \rW_{k})$ 
\end{assum}

\begin{assum}[Sequential Exchangeability (SE)]
    \label{assum:SE}
For $k = 0, \ldots, J-1$, 
	\[Y_J^{\rw_{J-1}} \indep W_k | \rW_{k-1}, \rX_{k}.\]
    
\end{assum}

Under the above assumptions, the g-formula \citep{gill2001causal} provides an identification of $\E[Y_J^{\rw_{J-1}}]$ given by
\begin{equation}
    \begin{aligned}
    \E[{Y}_J^{\rw_{J-1}}] = \int_{x_0}\ldots\int_{x_{J-1}} &\E[Y_J|\rX_{J-1} = \rx_{J-1}, \rW_{J-1} = \rw_{J-1}] \\ &\times\prod_{k = 0}^{J-1}\Pr(X_k \in \dx{x}_k|\rX_{k-1} = \rx_{k-1}, \rW_{k-1} = \rw_{k-1}).
    \end{aligned}
    \label{eq:gforx}
\end{equation}


\section{Continuous-time treatment processes}

\subsection{Setting and notation}

Define a standard probability space $(\Omega, \mathscr{F}, \Prob)$, where $\Omega$ is the sample space, $\fF$ is a $\sigma$-algebra, and $\Prob$ is a probability measure. Following \citet{sun2022causal}, we consider a finite study period which ends at time $T$, and define the index set $\mathcal{T} = [0, T]$.  Define the following stochastic processes: the treatment process $W: \Omega \times \mathcal{T}\rightarrow \R$, the outcome process $Y: \Omega \times \mathcal{T}\rightarrow \R$, and the covariate process $Z = (Z^1, \ldots, Z^{p-1}): \Omega \times \mathcal{T}\rightarrow \R^{p-1}$. 
We then define $X$ as a multivariate process, where $X = (X^1, \ldots, X^p) \equiv (Y, Z^1, \ldots, Z^{p-1})$, with superscripts from 1 to $p$ representing each component. To describe the treatment trajectory up to time $t$, we define $\rW(\omega, t) = \{W(\omega, s): 0 \le s \le t\}$, where $W(\omega, s)$ (or $W_s(\omega)$, whichever is convenient) is the value of the observable treatment trajectory at time $s$. 
Similarly, we define $\rX(\omega, t) = \{X(\omega, s): 0 \le s \le t\} = (\rY(\omega, t), \rZ(\omega, t))$.  We denote the observable full trajectories as $\mathcal{D} \equiv (\rX_T, \rW_T)$.

We define the ``potential outcome trajectory up to time $t$" as the values of the outcome and covariates that would have been observed up to time $t$, had the predetermined treatment plan of interest been precisely followed. This is denoted by 
\[\rX^{\rw(T)}(t) = \{X^{\rw(T)}(s): 0 \le s \le t\} = (\rY^{\rw(T)}(t), \rZ^{\rw(T)}(t)).\]
We assume that the future values of treatment cannot affect the past values of the outcome and covariates, so we have $\rX^{\rw(T)}(t) = \rX^{\rw(t)}(t)$.
Our primary interest is in estimating the ``counterfactual mean" of the treatment plan $\rw(T)$ at the end of the study, denoted by $\E[Y^{\rw}_T] \equiv \E[Y^{\rw(T)}(T)]$.

To connect these concepts with the observable data, we consider a ``deterministic discretization grid" $\Delta^J \equiv \{t_0, t_1, \ldots, t_J\},$ with $0 = t_0 < t_1< \ldots< t_J = T$. This allows us to relate the continuous-time data to the discrete-time data, by defining $W_{k} \equiv W(t_k)$ and $X_{k} \equiv X(t_k)$. Here, we use subscripts like ``$i, k, J$" for discrete-time data, and subscripts like ``$s, t, t_k, T$" for continuous-time data. Then, $\mathcal{D}^J \equiv \{\rW_J, \rX_J\}$ is the corresponding discrete-time observable data under $\Delta^J$.


\subsection{Discretization and identification assumptions}

Through the continuous-time lens, conceptual and analytical issues arise from discretizations under the potential outcome framework. In particular, certain identification assumptions that are relied upon in discrete-time g methods may be violated in continuous-time DGPs. For instance, unobserved trajectory values that impact both future treatment and outcome could violate the Sequential Exchangeability condition (Assumption \ref{assum:SE}). This scenario is illustrated in Figure \ref{fig:vioa}.

Moreover, in the continuous-time context, a single discrete-time treatment plan $\rw_{J-1}$ may correspond to multiple distinct treatment trajectories $\rw_T$, all of which pass through the same points $(t_k, w_k), k = 0, \ldots, J-1$. This situation results in multiple versions of treatment, thereby violating the Sequential Consistency condition (Assumption \ref{assum:SC}). As a result, the notation $Y^{\rw_{J-1}}_J$ becomes ill-defined. Figure \ref{fig:viob} illustrates this scenario.

\begin{figure}
    \centering
    
    \begin{subfigure}{1\textwidth}
      \centering
      \includegraphics[width=\linewidth]{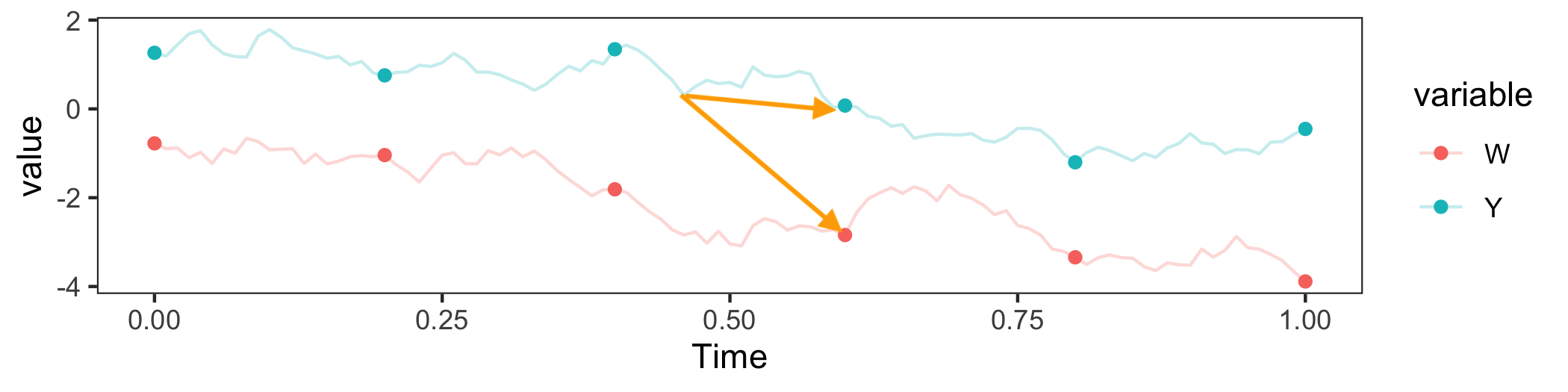}
      \caption{Violation of Sequential Exchangeability.}
      \label{fig:vioa}
    \end{subfigure}
    \begin{subfigure}{1\textwidth}
      \centering
      \includegraphics[width=\linewidth]{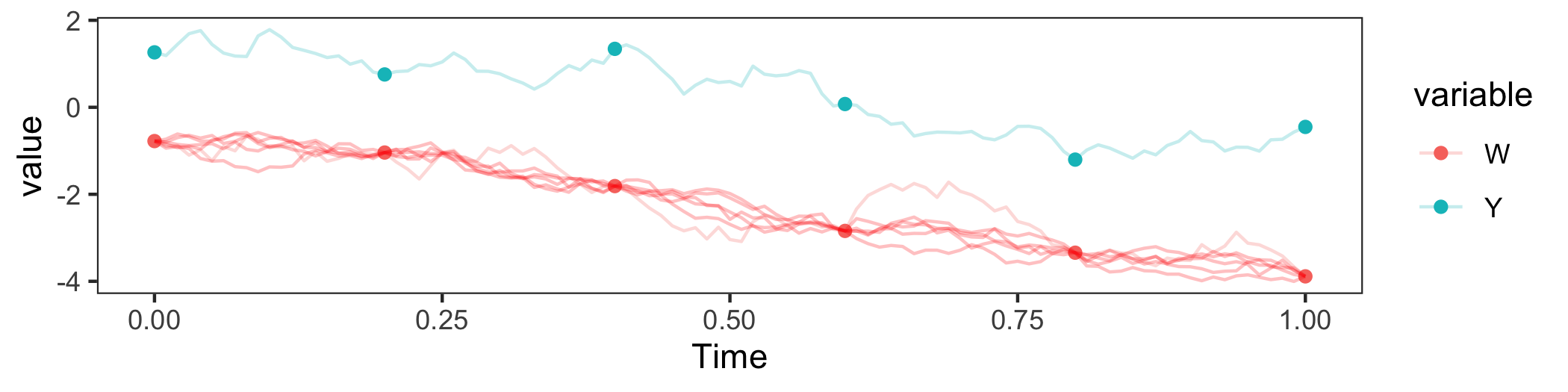}
      \caption{Violation of Sequential Consistency.}
      \label{fig:viob}
    \end{subfigure}
    
    \caption{Violations of discrete-time identification assumptions under continuous-time DGPs. The top panel shows how trajectory values between observed time points can impact observed $W$ and $Y$; the bottom panel shows how many possible treatment trajectories are compatible with the observed data and each of these may exert a different impact on the outcome. As a result, there are ``multiple versions of treatment'' across time points, and possibly across units. ($W$: treatment; $Y$: outcome; Dots are observed values on a discretization grid; Solid lines are full trajectories; Arrows show the impact of unobserved values on the treatment and outcome.)}
    \label{fig:vio}
\end{figure}

\subsection{Bias due to discretization}

Under continuous-time DGPs, standard discrete-time identification assumptions may not hold, leading to differences between g-method estimands and the true counterfactual mean. Specifically, for a given static treatment trajectory $\rw_T$ (or $\rw$ for simplicity), the true causal estimand is defined as \[\eta \equiv \E[Y_T^{\rw}].\]
For a given discretization grid $\Delta^J$, we denote the estimand from a generic discrete-time causal identification strategy as $\theta_{\Delta^J}$, simplified as $\theta_J$ for convenience. We define the mesh of a grid as $\abs{\Delta^J} \equiv \max_{i = 1, \ldots, J}(t_i - t_{i-1})$ to measure its density. Under any sequence of increasingly dense grids $\Delta^J$ with $\abs{\Delta^J} \rightarrow 0$ as $J \rightarrow \infty$, we define $\theta_\infty \equiv \lim_{J\rightarrow\infty}\theta_J$, assuming the existence of this limit.

The ``identification bias" is the difference \[\delta_J = \theta_J - \eta.\] It has two constituent parts: the ``discretization bias'' and the ``asymptotic bias". The discretization bias describes the difference between the finite $J$ estimand and the limit estimand, quantifying the influence of measurement density. It converges to 0 by definition. The asymptotic bias is the remaining bias, which can result from unmeasured confounding trajectories or improper adjustment strategies. For instance, the na\"ive adjustment strategy's estimand converges to the factual mean instead of the counterfactual mean when treatment-confounder feedback exists, leaving non-zero asymptotic bias (see Remark \ref{rem:naive} for analytical details). By decomposing the identification bias into asymptotic and discretization biases, we can better understand its causes and work to reduce it.
\[
    \underbracket{\delta_J}_{\text{Identification Bias}} = \underbracket{\left\{\theta_J - \theta_\infty\right\}}_{\text{Discretization Bias}}+ \underbracket{\left\{\theta_\infty  - \eta\right\}}_{\text{Asymptotic Bias}}. 
\]


\section{Identification for a class of continuous-time stochastic linear DGPs}
\label{sec:OU}

In this section, we address the problem of identifying the counterfactual mean in the context of continuous-time linear DGPs, which are commonly used to model time-varying systems where the underlying process is assumed to be both continuous and linear over time. We show that the g-formula, a widely used identification method in discrete-time settings, has zero asymptotic bias in this class of DGPs. This means that as the number of time points in the data increases, the g-formula provides an increasingly accurate estimand approaching the true counterfactual mean, making it a suitable choice for analyzing dense longitudinal data generated by continuous-time linear DGPs. The proofs of these results can be found in Appendix A.

While our focus is on linear DGPs, this choice is motivated by a desire for simplification and clarity. By focusing on this class, we aim to provide a clear and intuitive presentation of our framework. Understanding the properties of identification methods under linear models is an essential step toward developing strategies that will work reliably in more complex and realistic models.

\subsection{Linear DGPs}

We consider a setting where both the treatment process $(W_t)$ and the outcome process $(Y_t)$ are both one-dimensional processes, and $\mathcal{D} = \{\rW_T, \rY_T\}$. We specify the underlying data generating process (DGP) using stochastic differential equations, which are widely used in fields such as physics, biology, and finance \citep{bonilla2019active, rohlfs2014modeling, leung2015optimal}:
\begin{equation}
    \dx{\begin{pmatrix*} Y_t \\ W_t \end{pmatrix*}} = -\boldsymbol{\beta}\begin{pmatrix*} Y_t \\ W_t \end{pmatrix*}\dx{t} + \boldsymbol{\sigma}\dx{\boldsymbol{B}_t},
    \label{eq:2dOU}
\end{equation}
with initial values $(Y_0,  W_0)^{\intercal}$. Here, $\boldsymbol{\beta}$ and $\boldsymbol{\sigma}$ are constant $2 \times 2$ matrices, and $\boldsymbol{B}_t$ is a 2-dimensional standard Brownian motion that is independent of $(Y_0, W_0)^{\intercal}$. We denote $\boldsymbol{B}_t = (B^{1}_{t}, B^{2}_{t})^{\intercal}, \boldsymbol{\beta} = \begin{pmatrix*} \beta_{11}, & \beta_{12} \\  \beta_{21}, & \beta_{22} \end{pmatrix*}$, and $\boldsymbol{\sigma} = \begin{pmatrix*} \sigma_{11}, & \sigma_{12} \\  \sigma_{21}, & \sigma_{22} \end{pmatrix*}$. Note that the process $Y$ is also a time-varying confounder.

Assumption \ref{assum:treat} below is a weak regularity condition that ensures that the treatment plan is well-behaved:
\begin{assum}
    The treatment function of interest $w^*: \mathcal{T} \rightarrow \R$ is Borel-measurable,  Riemann-integrable, and bounded on $\mathcal{T}$. 
    \label{assum:treat}
\end{assum}

The potential outcome process under a plan $\rw^*(T)$ that satisfies Assumption \ref{assum:treat} is the solution of the following time-inhomogeneous linear stochastic equation:
\begin{equation}
    \dx{Y}_t^{\rw^*} = -(\beta_{11}Y_t^{\rw^*} + \beta_{12}w^*_t)\dx{t} + \sigma_{11}\dx{B^1}_{t} + \sigma_{12}\dx{B^2}_{t},
    \label{eq:1dPOOU}
\end{equation}
where ${Y}_0^{\rw^*} \equiv Y_0$, since no treatment has occurred yet. For the above DGPs, we can solve exactly for the mean potential outcome. By \citet[Chapter~4]{gardiner1985handbook}, our estimand of interest is 
\[\eta \equiv \E[Y_T^{\rw^*}] = e^{-\beta_{11}T}\E[Y_0] - \beta_{12}\int_0^T w^*_s e^{\beta_{11}(s-T)}\dx{s}.\] 
Note that the term $(-\beta_{12})$ can be interpreted as the instantaneous effect of treatment on the outcome, and the integral term represents the cumulative effect of treatment over time.

The causal DAG in Figure \ref{fig:DAG} helps to illustrate the causal relationships between the treatment, outcome, and confounder processes in this DGP. The arrows indicate the direction of causality, and the coefficients on the arrows quantify the strength of this relationships.  
Figure \ref{fig:traj} shows simulated trajectories of the observable processes, as well as the potential outcome process under a constant treatment plan of $w_t^* \equiv 1$ for all $t$, with parameter specifications $\boldsymbol{\beta} = \begin{pmatrix*} 0.2, & -5\\ -3, & 0.5 \end{pmatrix*}, \boldsymbol{\sigma} = \begin{pmatrix*} 1, & 0.3\\ 0.3, & 0.5\end{pmatrix*}$.

\begin{figure}
    \centering
    \begin{subfigure}[t]{0.45\textwidth}
        \centering
        \includegraphics{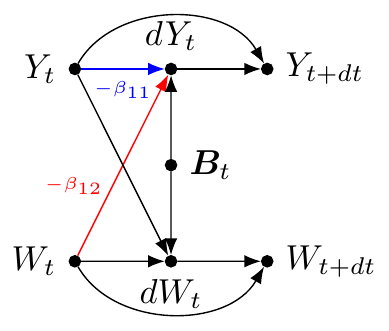}
        \caption{Observational process DAG} \label{fig:DAG1}
    \end{subfigure}
    \hfill
    \begin{subfigure}[t]{0.45\textwidth}
        \centering
        \includegraphics{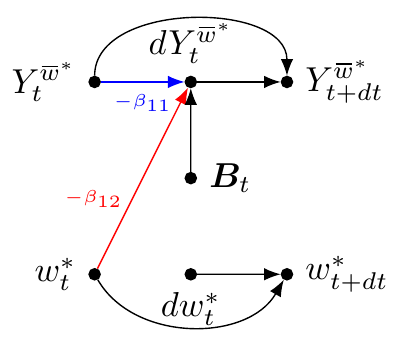}
        \caption{Counterfactual process DAG under a deterministic intervention $\rw^*$} \label{fig:DAG2}
    \end{subfigure}
    \caption{Causal directed acyclic graphs (DAG) representing the continuous-time linear DGPs  in an infinitesimal time interval $[t, t+\text{d}t]$. Panel (a) showcases the observational process where $W$ and $Y$ have feedback between each other. Panel (b), on the other hand, portrays the counterfactual process under a predetermined deterministic intervention, represented by $\rw^*$. In this panel, the arrows from $Y$ and $\mathbf{B}$ leading into $W$ are eliminated, due to the intervention. The red arrows present in both graphs symbolize the instantaneous effect of the treatment on the outcome, expressed as $-\beta_{12}$.}
    \label{fig:DAG}
\end{figure}

\begin{figure}
    \centering
    \includegraphics[width = 0.85\textwidth]{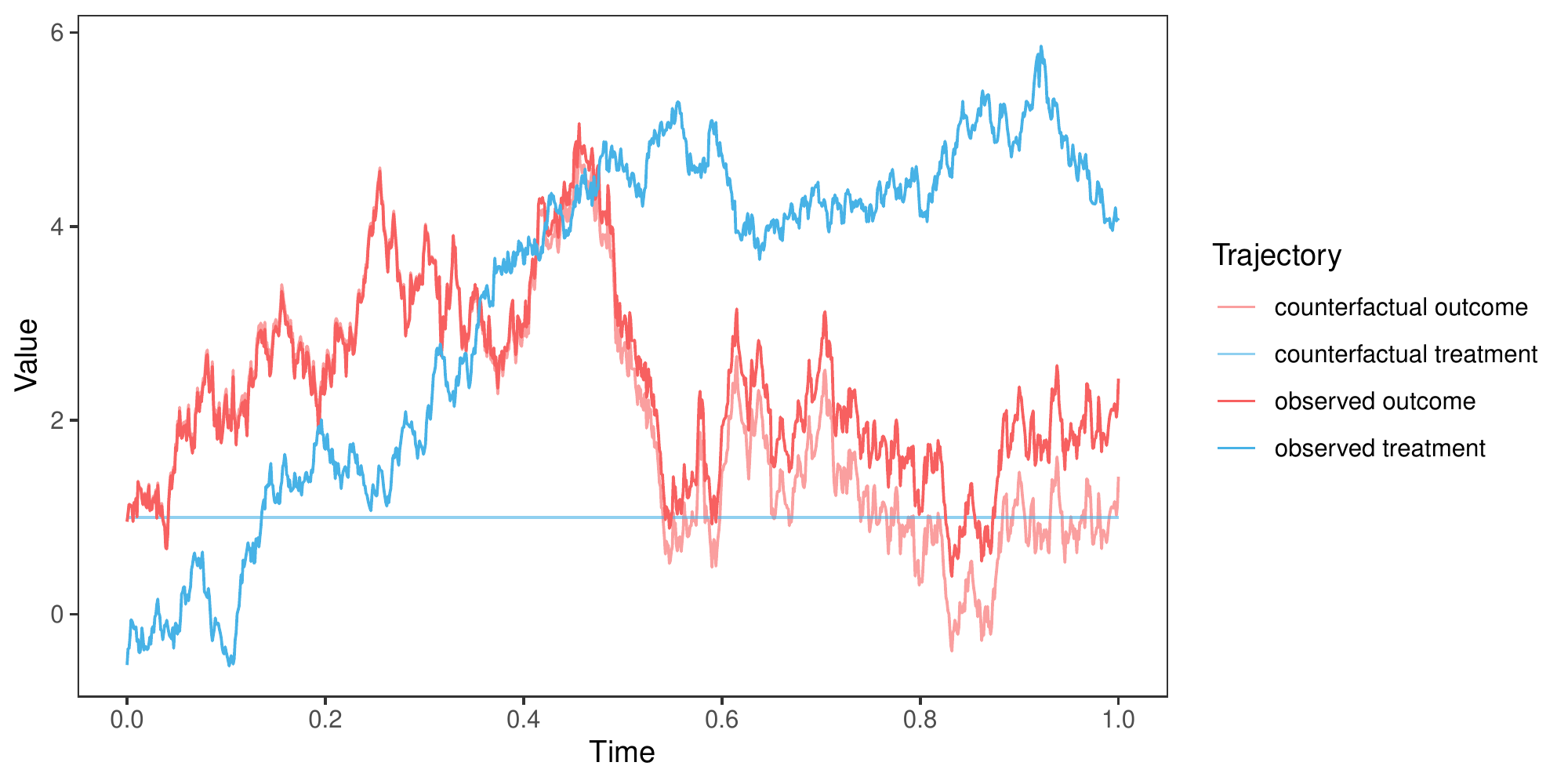}
    \caption{ Realizations of both observed and counterfactual trajectories under a predetermined constant treatment plan, denoted as $w_s^* \equiv 1$. Observed trajectories refer to the data collected under the original conditions, reflecting the natural course of events. On the other hand, counterfactual trajectories represent hypothetical scenarios under the constant treatment plan, showing what could have been the outcome if the treatment plan had been implemented consistently. This visualization aids in comparing the potential impact of maintaining a treatment strategy versus the outcomes observed under actual conditions.}
    \label{fig:traj}
\end{figure}

\subsection{Identification bias}

We will begin by considering the equidistant discretization grid $\Delta^J$, where $t_k = kT/J$. Suppose we have regular longitudinal data $\mathcal{D}^J$ on $\Delta^J$ under the DGP specified in Equation \eqref{eq:2dOU}. We apply the g-formula to the discrete-time data to investigate its identification bias. We refer to the g-formula estimand as $\theta^g$, and we define $\boldsymbol{\gamma}(k) = \begin{pmatrix*} \gamma_{11}(k) & \gamma_{12}(k) \ \gamma_{21}(k) & \gamma_{22}(k) \end{pmatrix*} \equiv e^{-\boldsymbol{\beta}\frac{T}{k}}$, where the right-hand side is a matrix exponential. Using Theorem \ref{thm:bias}, we can explicitly characterize $\delta_J$ of $\theta^g$.

\begin{thm}
    Under Assumption \ref{assum:treat} and the treatment plan $\rw^*_T$, the identification bias of the g-formula $\theta^g_J$ is
    \begin{equation}
        \delta_J = \left(\gamma_{11}^J(J) - e^{-\beta_{11}T}\right)\E[Y_0] + \gamma_{12}(J)\left(\sum_{i = 0}^{J-1}w^*_i \gamma_{11}^{J-i-1}(J)\right) + \beta_{12}\int_0^T w^*_s e^{\beta_{11}(s-T)}\dx{s}.
        \label{eq:bias}
    \end{equation}
    \label{thm:bias}
\end{thm}

Theorem \ref{thm:bias} shows that the identification bias depends on the number of repeated measurements, the causal magnitude $\boldsymbol{\beta}$, the treatment plan $\rw^*$, and the study length $T$. In Figure \ref{fig:bias}, we show the identification bias under different parameter specifications given the constant treatment plan $w_s^* \equiv 1$. Note that the identification bias is generally non-zero for any given $J$, except in special cases, e.g., when the sharp causal null hypothesis is true, as stated in Corollary \ref{cor:sharp_null}. Figure \ref{fig:bias} also suggests that $\delta_J$ goes to zero when $J$ increases. We formally show this result in Theorem \ref{thm:asymp_unbias}, which indicates that the g-formula is asymptotically unbiased for the linear DGPs specified above.

\begin{cor}
    Under the setting in Theorem \ref{thm:bias}, $\delta_J = 0$ for all $J$ if $\beta_{12} = 0$.
    \label{cor:sharp_null}
\end{cor}

\begin{figure}
    \centering
    \includegraphics[width = 0.7\textwidth]{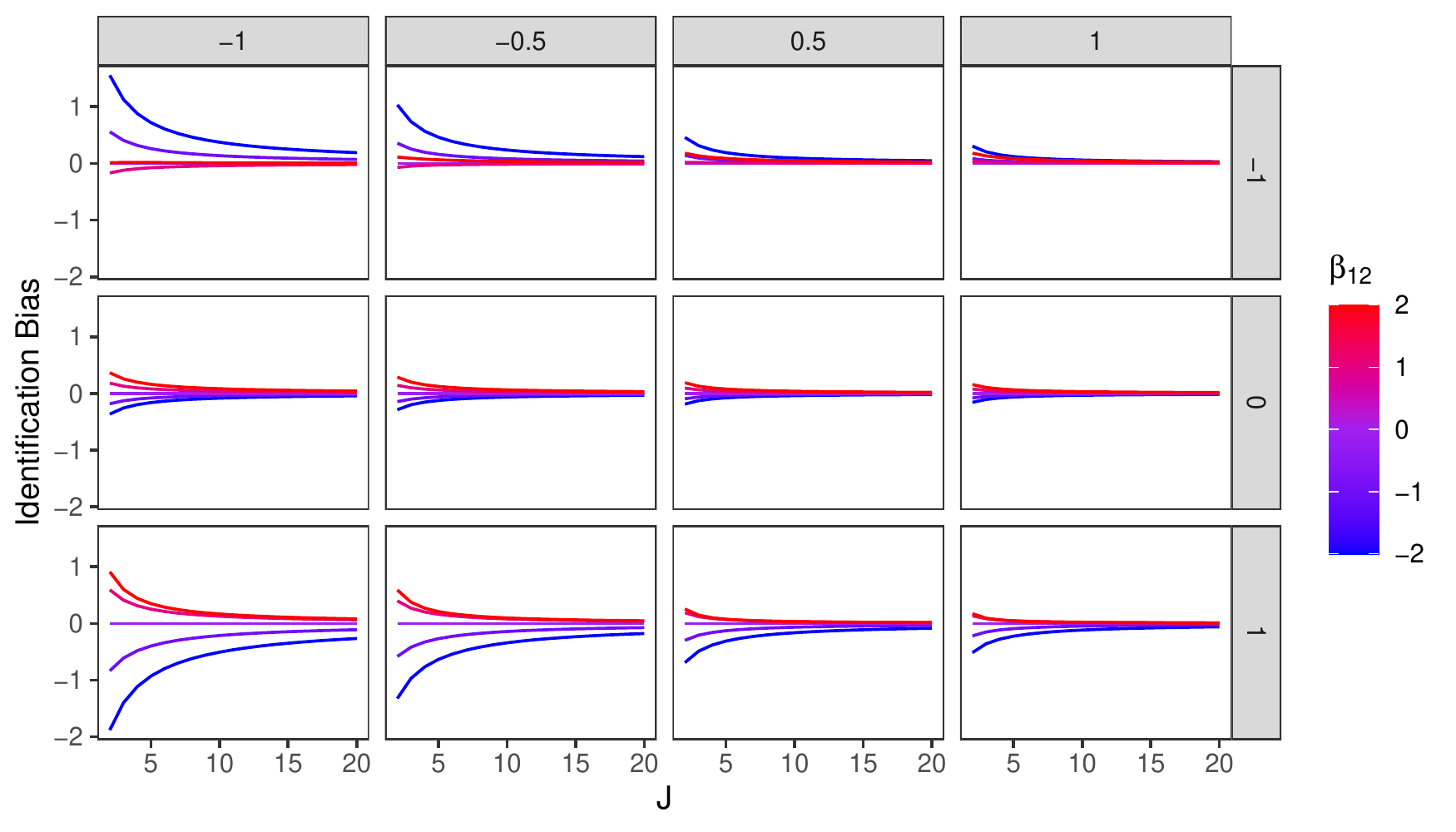}
    \caption{The figure illustrates the identification bias of the g-formula, denoted as $\theta^g_J - \eta$, under the framework of linear Data-Generating Processes (DGPs). The bias is presented in a grid format, where each row and column corresponds to different values of $\beta_{21}$ and $\beta_{11}$ respectively. These values represent specific parameters within the linear DGP model. The causal parameter, designated as $\beta_{12}$, is set at varying levels, specifically $-2, -1, 0, 1, 2$. This figure provides a visual representation of how changes in the parameters $\beta_{21}$, $\beta_{11}$, and $\beta_{12}$ can affect the identification bias of the g-formula in the context of linear DGPs. In particular, they suggest that the identification bias tends to zero as $J$ grows. (Other parameters in the model have been assigned specific values for the purpose of this illustration: $T = 1, w_s^* \equiv 1, \beta_{22} = 0.5, \E[Y_0] = \E[Y_0^{\rw^*}] = 1$.)}
    \label{fig:bias}
\end{figure}

\begin{thm} 
    Under the setting in Theorem \ref{thm:bias},
    \[\lim_{J \rightarrow \infty} \theta^g_J = \eta.\]
    \label{thm:asymp_unbias}
\end{thm}

\begin{remark}{(Identification bias of na\"ive adjustments.)}
    With the same data $\mathcal{D}_J$, the na\"ive adjustment strategy $\tilde{\theta}$ assumes $Y_J^{\rw^*} \indep \rW_{J-1}|\rY_{J-1}$, and $\tilde{\theta}_J = \E\E[Y_J|\rY_{J-1}, \rW_{J-1} = \rw^*_{J-1}]$. Straightforward calculations give 
    \begin{equation*}
        \begin{aligned}
            \tilde{\theta}_J &= \gamma_{12}(J)w_{J-1}^* + \gamma_{11}(J)\left[\gamma_{11}\left(\frac{J}{J-1}\right)\E(Y_0) + \gamma_{12}\left(\frac{J}{J-1}\right)\E(W_0)\right],\\
            \tilde{\theta}_\infty &\equiv \lim_{J \rightarrow \infty} \tilde{\theta}_J = \gamma_{11}(1)\E(Y_0) + \gamma_{12}(1)\E(W_0) = \E[Y_T] \ne \E[Y_T^{\rw^*}],
        \end{aligned}
    \end{equation*}
    showing that na\"ive adjustments are \emph{not} asymptotically unbiased.
    \label{rem:naive}
\end{remark}

\begin{remark}{(Confidence intervals under finite sample.)}
In practice, given a finite sample size $n$, an estimator for an asymptotically unbiased identification strategy will often have decreasing bias as $J$ increases. However, confidence intervals are often used to quantify the statistical uncertainty for the estimate, and in the presence of discretization bias can be misleading, with the actual coverage probability for the true estimand much lower than the nominal probability. This emphasizes the importance to study discretization and novel bias-corrected confidence intervals with validity.
\end{remark}


\section{Heuristic sensitivity analysis on discretization}

To entirely eliminate discretization bias for variables that evolve continuously over time, it is essential to have complete trajectory datasets. However, depending on the variable's volatility, a sufficiently large number of measurements, $J_0$, can capture crucial variations while introducing only negligible discretization bias. When empirical knowledge for determining $J_0$ is lacking, conducting sensitivity analysis for $J$ is vital to ensure that conclusions do not rely unduly on specific discretization choices. \citet{lange2019commentary} proposed a heuristic sensitivity analysis method, suggesting that if $J$ is sufficiently large, a marginally smaller $J'$ would produce similar estimates. Thus, comparing estimates using $J$ and $J'$ can indicate whether $J$ is large enough. It is important to note, however, that according to our proposed framework, this approach's effectiveness depends on the implicit assumption that the identification strategy functional, $\theta_J$, has a finite limit under the continuous-time data-generating process (DGP), as estimate stability would otherwise be unachievable. 

Determining whether the estimates change ``substantially'' when comparing them under different discretization scales can be somewhat subjective. In the following, we propose a simple and heuristic quantitative measure for sensitivity concerning discretization. We will demonstrate the proposed approach through simulations.

\subsection{Discretization sensitivity measure $\zeta$}

We introduce a discretization sensitivity measure, denoted as $\zeta$. This measure is designed to evaluate the potential impact of discretization bias on the estimated effects within a study. The value of $\zeta$ provides an indication of how strongly the potential discretization bias could potentially alter or ``explain away'' the observed effects in the study. A low value would suggest that the observed effects are highly sensitive to discretization bias. Moreover, by providing a common metric, the $\zeta$ measure facilitates the comparison and discussion of sensitivity across different studies, enhancing the interpretability and comparability of research findings. The formal definition is as follows.

Suppose that the estimand of interest, $\tau$, is the expectation of a contrast, such as the contrast between two potential outcomes. Specifically, the null hypothesis is $H_0: \tau = 0$. Under a finite $J$, the corresponding discrete-time estimand is denoted as $\tau_J$. Let $\hat\tau_J$ be the point estimate and $[\hat\tau_J^l, \hat\tau_J^u]$ be the $1-\alpha$ confidence interval (CI) for $\tau_J$. To address the issue of discretization bias, we propose a sensitivity measure, $\zeta$, defined as follows:
\[\zeta \equiv 
\begin{cases}
    \frac{\min\{|\hat\tau_J^l|, |\hat\tau_J^u|\}}{|\hat\tau_J - \hat\tau_{J/2}|} & \text{if } 0 \notin [\hat\tau_J^l, \hat\tau_J^u]\\
    0 & \text{otherwise} \\
\end{cases},
\]
Here, $\hat\tau_{J/2}$ is the point estimate for $\tau_{J/2}$ obtained by applying the original analysis plan using half of the time points equidistantly. Heuristically, a larger $\zeta$ implies stronger evidence supporting the alternative hypothesis $\tau \ne 0$. 
This is because with a sufficiently large $J$, $|\hat\tau_J - \hat\tau_{J/2}|$ will be close to zero. The denominator $|\hat\tau_J - \hat\tau_{J/2}|$ serves as a proxy for the magnitude of discretization bias for $J$ and normalizes the current ``effect size''. Note that $\zeta$ is intentionally designed to inform possible Type I error. We suggest that researchers report point and interval estimates, along with $\zeta$, when discretization poses a potential threat to the credibility of the results. Researchers should be cautious when $\zeta$ is small.

\subsection{Simulations}

To demonstrate the behavior of $\zeta$ at varying discretization scales and parameter values, we conducted simulations using the continuous-time DGP specified by Equation \eqref{eq:2dOU}. The parameters were set as follows:
\[ \boldsymbol{\beta} = \begin{pmatrix*} 0.2, & \beta_{12} \\  -3, & 0.5 \end{pmatrix*}, \boldsymbol{\sigma} = \begin{pmatrix*} 1, & 0.3 \\  0.3, & 0.5 \end{pmatrix*},  \begin{pmatrix*} Y_0 \\ W_0\end{pmatrix*} \sim \text{Normal}\left( \begin{pmatrix*} 1 \\ 0\end{pmatrix*}, \begin{pmatrix*} 0.25, & 0 \\  0, & 0.25 \end{pmatrix*} \right).\]
We varied the instantaneous effect $\beta_{12}$ from $-10$ to $-3$. The study ended at $T = 1$, and we considered a treatment plan of interest, $w_t^* \equiv 1$, and a baseline treatment, $w_t^\circ \equiv 0$. The true causal estimand was $\tau = \E[Y^{\rw^*}_T] - \E[Y^{\rw^\circ}_T]$. We recorded $J+1$ equally spaced time points in the study, with $J$ ranging from 4 to 40. The analysis included a sample of $n = 200$ individuals. To model $\E[Y_k|Y_{k-1}, W_{k-1}]$, we employed a linear model with the specification $ Y_k \sim Y_{k-1} + W_{k-1}$. The study's confidence level was set at $1- \alpha = 0.95$, and we used $N = 500$ bootstrap samples to estimate the confidence intervals.
The results, depicted in Figure \ref{fig:zeta}, showed that $\zeta$ tended to increase steadily with $J$, although with small fluctuations attributable to finite sample variability.

\begin{figure}
    \centering
    \includegraphics[width = 0.8\linewidth]{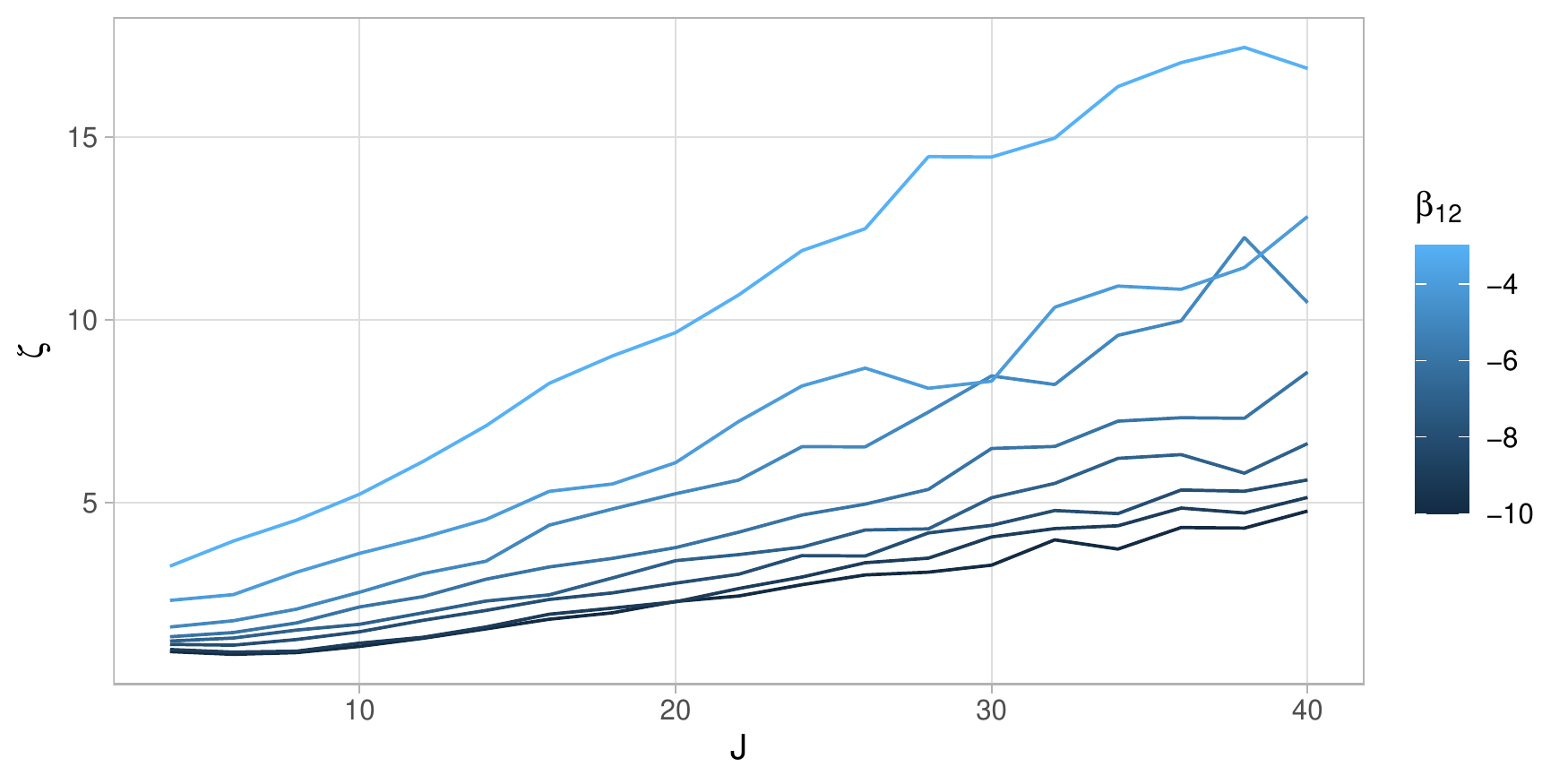}
    \caption{A visual representation of simulation results, illustrating how the discretization sensitivity measure, represented as $\zeta$, changes with different values of $J$ and the causal parameter $\beta_{12}$. The measure $\zeta$ describes the extent to which potential discretization bias can explain away estimated effects, and facilitates comparisons and discussions of sensitivity across studies. The figure demonstrates that $\zeta$ generally exhibits an increasing trend with an augmentation in $J$.}
    \label{fig:zeta}
\end{figure}

\section{Discussion}

As standard errors in large-scale longitudinal studies continue to decrease, it is increasingly important to identify and address potential sources of bias. In this paper, we investigate discretization bias and asymptotic bias in observational studies using g-methods under continuous-time data-generating processes (DGPs). Given the complexity of the general problem, it is essential to analyze specific instances that can yield at least partial progress. Therefore, we focus on linear DGPs and find that the g-formula is an effective identification method for counterfactual means in such DGPs, providing increasingly accurate estimands as the number of measurements grows. With the advancement of technology and the decreasing cost of measurements, this type of dense data is becoming more available.

Discretization is related to the broader literature on missing data and coarsening \citep{tsiatis2006semiparametric}, and should be handled with equal care. We recommend that researchers incorporate an analysis plan for discretization in their protocol prior to the primary analysis and conduct sensitivity analysis afterward.

Several open questions remain, including:
\begin{enumerate}
    \item Without imposing strict assumptions on the underlying continuous-time DGP, identification bias for discrete-time methods is generally inevitable. The only apparent solution to reduce identification bias to zero is to gather denser observations. Therefore, understanding when asymptotic bias is zero is crucial for statisticians and empirical researchers. We demonstrate that the g-formula exhibits zero asymptotic bias for a class of commonly used linear DGPs. However, asymptotically unbiased identification strategies for nonlinear and/or non-Markovian DGPs are also of significant interest. Recent work on causal identification approaches that address continuous-time DGPs directly may offer theoretical foundations and practical tools to answer this question \citep[e.g.,][]{sun2022causal}.
    \item Many electronic health record-based longitudinal studies feature irregular observations, which introduce both discretization bias and potential selection bias due to informative visits \citep{harton2022informative, pullenayegum2022randomized}. Further research in this area is necessary and highly relevant to real-world situations. 
    \item The number of subjects ($n$) and the number of repeated measurements ($J$) have a significant impact on the bias and variance of the estimator. In situations where resources are limited, developing new longitudinal designs that achieve an    ``optimal" bias-variance tradeoff would be highly advantageous for researchers.
\end{enumerate}

\textbf{Acknowledgements}: We are grateful to Fan Li and P. M. Aronow for helpful comments and discussion. This work was supported by NIH grant 1DP2HD091799-01.


\section*{Appendix A: Proofs}

\subsection*{Theorem \ref{thm:bias}}
\begin{proof}[Proof of Theorem \ref{thm:bias}.]
    It is a fact \citep[Chapter~4.4.6]{gardiner1985handbook} that
    \begin{equation}
        \E\left[\begin{pmatrix*} Y_t \\ W_t \end{pmatrix*}\right] = e^{-\boldsymbol{\beta}t}\E\left[\begin{pmatrix*} Y_0 \\ W_0 \end{pmatrix*}\right],
        \label{eq:2dexpected}
    \end{equation}
    where $e^{-\boldsymbol{\beta}t}$ is a matrix exponential, and (c.f. \citep[Chapter~4.4.7]{gardiner1985handbook})
    \begin{equation}
        \eta =  e^{-\beta_{11}T}\E[Y_0] - \beta_{12}\int_0^T w^*_s e^{\beta_{11}(s-T)}\dx{s}.
        \label{eq:true_estimand}
    \end{equation}
    
    For treatment plan $\rw^*_{J-1} = (w^*(t_0), \ldots, w^*(t_{J-1}))$, by g-formula Equation \eqref{eq:gforx}, we have
    \begin{equation}
        \begin{aligned}
        \theta^g_J = \int_{y_0}\ldots\int_{y_{J-1}} &\E[Y_J|\rY_{J-1} = \ry_{J-1}, \rW_{J-1} = \rw^*_{J-1}] \\ &\times\prod_{k = 0}^{J-1}\Pr(Y_k \in \dx{y}_k|\rY_{k-1} = \ry_{k-1}, \rW_{k-1} = \rw^*_{k-1}).
        \end{aligned}
        \label{eq:gforz}
    \end{equation}
    Then, by the Markov property, temporal homogeneity, and Equation \eqref{eq:2dexpected}, 
    \[\E[Y_k|\rY_{k-1}, \rW_{k-1}] = (1,0)^{\intercal}e^{-\boldsymbol{\beta}\frac{T}{J}}\begin{pmatrix*} Y_{k-1} \\ W_{k-1} \end{pmatrix*}.\] 
    Therefore, 
    \begin{equation}
        \E[Y_k|\rY_{k-1}, \rW_{k-1} = \rw^*_{k-1}] = \gamma_{11}(J)Y_{k-1} + \gamma_{12}(J)w^*_{k-1}.
        \label{eq:update}
    \end{equation}
    Plug in Equation \eqref{eq:update} to \eqref{eq:gforz} iteratively, then we have the estimand
    \[ \theta^g_J  =\gamma_{11}^J(J)\E[Y_0] + \gamma_{12}(J)\left(\sum_{i = 0}^{J-1}w^*_i \gamma_{11}^{J-i-1}(J)\right). \]
    Contrasted the above with Equation \eqref{eq:true_estimand}, we hence finished the proof.
\end{proof}

\subsection*{Corollary \ref{cor:sharp_null}}

\begin{proof}[Proof of Corollary \ref{cor:sharp_null}.]
    Suppose the two eigenvalues of $\boldsymbol{\beta}$ are $\lambda_1, \lambda_2$. Then by Cayley-Hamilton theorem in linear algebra, 
    \[e^{t\boldsymbol{\beta}} = s_0(t)I + s_1(t)\boldsymbol{\beta},\]
    where $s_0(t) = \frac{\lambda_1e^{\lambda_2 t} - \lambda_2e^{\lambda_1 t}}{\lambda_1 - \lambda_2}, s_1(t) = \frac{e^{\lambda_1 t} - e^{\lambda_2 t}}{\lambda_1 - \lambda_2}$ when $\lambda_1 \neq \lambda_2$; when $\lambda_1 = \lambda_2$, $s_0(t) = (1-\lambda_1 t)e^{\lambda_1 t}, s_1(t) = te^{\lambda_1 t}$. Then, $\forall \lambda_1, \lambda_2$ (real or complex), we have $\gamma_{11}(J) = s_0(-T/J) + s_1(-T/J)\beta_{11},$ $\gamma_{12}(J) = s_1(-T/J)\beta_{12}$.

    Since $\beta_{12} = 0$, we have $\lambda_1 = \beta_{11}, \lambda_2 = \beta_{22}$. Then, it is obvious that $\gamma_{12}(J) = 0$.\\
    (i) When $\beta_{11} \ne  \beta_{22}$,
    \begin{equation*} 
        \begin{aligned}
            \gamma_{11}(J) &= \frac{\lambda_1 e^{-T\lambda_2/J} - \lambda_2 e^{-T\lambda_1/J}}{\lambda_1 - \lambda_2} + \beta_{11}\frac{e^{-T\lambda_1/J } - e^{-T\lambda_2/J}}{\lambda_1 - \lambda_2}\\
        & = e^{-\beta_{11}T/J}.
        \end{aligned}
    \end{equation*}
    Thus, $\gamma_{11}^J(J) - e^{-\beta_{11}T} = 0$.

    (ii) When $\beta_{11} = \beta_{22}$, one can similarly show that $\gamma_{11}^J(J) - e^{-\beta_{11}T} = 0$, and we thus omit the details.

    Then, three terms in Equation \eqref{eq:bias} are all equal to 0, and thus $\delta_J = 0$. We hence finished the proof.
\end{proof}

\subsection*{Theorem \ref{thm:asymp_unbias}}
\begin{proof}[Proof of Theorem \ref{thm:asymp_unbias}.]
    Note that $\lim_{|t| \rightarrow 0}s_0(t) = 1, \lim_{|t| \rightarrow 0}s_1(t)  = 0, \frac{\dx{s_0(t)}}{\dx{t}}|_{t = 0} = 0, \frac{\dx{s_1(t)}}{\dx{t}}|_{t = 0}  = 1, \lim_{J \rightarrow \infty} Js_1(-\frac{T}{J}) = -T, \lim_{J \rightarrow \infty} \gamma_{11}(J) = 1.$
    
    Since 
    \begin{equation*}
        \begin{aligned}
    \left(\lim_{J \rightarrow \infty} \theta^g_J\right) -  \eta & = \left\{\E[Y_0](\lim_{J \rightarrow \infty} \gamma_{11}^J(J) - e^{-\beta_{11}T})\right\} \\ & + \left\{\lim_{J \rightarrow \infty}\gamma_{12}(J)\left(\sum_{i = 0}^{J-1}w_i^* \gamma_{11}^{J-i-1}(J)\right) + \beta_{12}\int_0^T w^*_s e^{\beta_{11}(s-T)}\dx{s} \right\},
        \end{aligned}
    \end{equation*}
    we consider the terms in the first and second braces separately.

(1) It is easy to verify that $\lim_{J \rightarrow \infty} J\log(\gamma_{11}(J)) = -\beta_{11}T, \forall \lambda_1, \lambda_2$, using L'Hospital's rule. Thus, $\lim_{J \rightarrow \infty} \gamma_{11}^J(J) = e^{-\beta_{11}T}$.

(2) Next, we define $\epsilon_J \equiv \lim_{J \rightarrow \infty}\gamma_{12}(J)\left(\sum_{i = 0}^{J-1}w_i^* \gamma_{11}^{J-i-1}(J)\right) + \beta_{12}\int_0^T w^*_s e^{\beta_{11}(s-T)}\dx{s}$. 

We have 
\begin{equation*}
    \begin{aligned}
        \lim_{J \rightarrow \infty}\epsilon_J 
        & \equiv \lim_{J \rightarrow \infty} J\gamma_{12}(J)\left(\frac{1}{J} \sum_{i = 0}^{J-1}w_i^* \gamma_{11}^{J-i-1}(J) \right) + \beta_{12}\int_0^T w^*_s e^{\beta_{11}(s-T)}\dx{s} \\
        & = \lim_{J \rightarrow \infty} (J\gamma_{12}(J)) \lim_{J \rightarrow \infty} \left(\frac{1}{J} \sum_{i = 0}^{J-1}w_i^* \gamma_{11}^{J-i-1}(J) \right) + \beta_{12}\int_0^T w^*_s e^{\beta_{11}(s-T)}\dx{s} \\
        & = \lim_{J \rightarrow \infty} (J\beta_{12}s_1(-\frac{T}{J})) \lim_{J \rightarrow \infty} \left(\frac{1}{J} \sum_{i = 0}^{J-1}w_i^* \gamma_{11}^{J-i-1}(J) \right) + \beta_{12}\int_0^T w^*_s e^{\beta_{11}(s-T)}\dx{s} \\
        & = -T\beta_{12} \lim_{J \rightarrow \infty} \left(\frac{1}{J} \sum_{i = 0}^{J-1}w_i^* \gamma_{11}^{J-i-1}(J) \right) + \beta_{12}\int_0^T w^*_s e^{\beta_{11}(s-T)}\dx{s} \\
        & = \beta_{12} \left[ - \lim_{J \rightarrow \infty} \left(\frac{T}{J} \sum_{i = 0}^{J-1}w_i^* \gamma_{11}^{J-i-1}(J) \right) + \int_0^T w^*_s e^{\beta_{11}(s-T)}\dx{s} \right]\\
        & = \beta_{12} \left[ - \lim_{J \rightarrow \infty} \gamma_{11}^J(J) \lim_{J \rightarrow \infty} \left(\frac{T}{J} \sum_{i = 0}^{J-1}w_i^* \gamma_{11}^{-i-1}(J) \right) + \int_0^T w^*_s e^{\beta_{11}(s-T)}\dx{s} \right]\\
        & = \beta_{12} \left[ - e^{-\beta_{11}T} \lim_{J \rightarrow \infty} \left(\frac{T}{J} \sum_{i = 0}^{J-1}w_i^* \gamma_{11}^{-i-1}(J) \right) + \int_0^T w^*_s e^{\beta_{11}(s-T)}\dx{s} \right]\\
        & = \beta_{12}e^{-\beta_{11}T} \left[ -  \lim_{J \rightarrow \infty} \left(\frac{T}{J} \sum_{i = 0}^{J-1}w_i^* \gamma_{11}^{-i-1}(J) \right) + \int_0^T w^*_s e^{\beta_{11}s}\dx{s} \right]\\
    \end{aligned}
\end{equation*}

Define $t_i = \frac{T}{J} i, f_J(x) = w^*(x)\gamma_{11}^{-\frac{J}{T}x - 1}(J)$. $f_J$ is Riemann-integrable if $w^*$ is. Then, 
\begin{equation*}
    \begin{aligned}
        \lim_{J \rightarrow \infty} \left(\frac{T}{J} \sum_{i = 0}^{J-1}w_i^* \gamma_{11}^{-i-1}(J) \right) 
        & =  \lim_{J \rightarrow \infty} \sum_{i = 0}^{J-1} \frac{T}{J} f_J(t_i) \\
        & = \int_0^T \lim_{J \rightarrow \infty} f_J(s) \dx{s} \text{ (By the properties of Riemann-integrable functions)} \\
        & = \int_0^T w^*(s)e^{\beta_{11}s}\dx{s}.
    \end{aligned}
\end{equation*}
Hence we have finished the proof.

\end{proof}

\bibliography{discretize}
\end{document}